%% file: main.tex
\documentclass[11pt]{article}

\usepackage{arxiv}
\renewcommand{\headeright}{A Preprint}
\renewcommand{\undertitle}{A Preprint}
\renewcommand{\shorttitle}{Minimality of Random Moore Automata}

\usepackage[T1]{fontenc}
\usepackage[utf8]{inputenc}
\usepackage{lmodern}
\usepackage{microtype}

\usepackage{amsmath,amssymb,amsthm,mathtools}
\usepackage{dsfont}

\usepackage{enumitem}
\usepackage[hidelinks]{hyperref}
\usepackage{cleveref}
\usepackage{color}

\input{macros}

\title{Minimality of Random Moore Automata under Prefix-Dependent Congruences}

\author{
Mat\'ias Carrasco\textsuperscript{*}\\
Universidad ORT Uruguay\\
Facultad de Ingenier\'ia\\
\url{carrasco_m@ort.edu.uy}
\And
Sergio Yovine\textsuperscript{*}\\
Universidad ORT Uruguay\\
Facultad de Ingenier\'ia\\
\url{yovine@ort.edu.uy}
}

\date{}

\begin{document}

\twocolumn[
\maketitle

\begin{abstract}
We study prefix-dependent congruences for random deterministic transition systems with state outputs. In this setting, the admissible continuations used to compare two states may depend on the observed prefix, and two states are identified only if no common admissible continuation distinguishes their future outputs. The framework includes probabilistic deterministic finite automata as a motivating special case. We analyze the random transition model in which all transition values are independent and uniform. Each state is also assigned an independent label that specifies both its output and its set of admissible symbols. If two independent labels agree with probability strictly less than one, and every label has at least three admissible symbols, then the induced congruence is trivial with high probability. The proof combines a pruning process on pairs, a collision-free exploration controlling its early evolution, and a first-moment argument showing that the remaining pairs cannot organize into nontrivial equivalence classes.
\end{abstract}


\vspace{2em}
]

\begingroup
\renewcommand{\thefootnote}{\fnsymbol{footnote}}
\footnotetext[1]{This work was supported by Agencia Nacional de Investigaci\'on e Innovaci\'on (ANII), grant FMV\_1\_2023\_1\_175864.}
\endgroup

\input{section1-introduction}
\input{section2-automata-and-congruences}

\input{section3-random-model}
\input{section4-pair-pruning}
\input{section5-stable-block-systems}
\input{section6-main-results}
\input{section7-conclusion}

\bibliographystyle{elsarticle-num}
\bibliography{references}

\end{document}

%% file: macros.tex
\newtheorem{theorem}{Theorem}
\newtheorem{proposition}[theorem]{Proposition}
\newtheorem{lemma}[theorem]{Lemma}
\newtheorem{corollary}[theorem]{Corollary}

\theoremstyle{definition}
\newtheorem{definition}[theorem]{Definition}

\theoremstyle{remark}
\newtheorem{remark}[theorem]{Remark}

\newcommand{\Adm}{\operatorname{Adm}}
\newcommand{\Obs}{\operatorname{Obs}}
\newcommand{\obs}{\operatorname{obs}}

%% file: section1-introduction.tex
\section{Introduction}
\label{sec:introduction}

In many applications, the study of algorithmic complexity and implementation performance is done on random or ``average-case'' automata. Examples include the average-case analysis of Moore's minimization algorithm on random deterministic automata \cite{bassino2009average,david2012average, nicaud2014random}; the use of random Mealy machines in black-box checking \cite{peled1999black} and in the evaluation of gray-box learning algorithms \cite{abel2016gray}. In the probabilistic setting, related performance evaluations arise in learning algorithms for PDFA \cite{carrasco2024analyzing}.

Learning algorithms \cite{angluin1987learning} typically aim to recover, or approximate, a minimal representation of the observed behavior. Minimality is therefore a central property in automata learning. When random automata are used as benchmark instances, minimality with high probability may raise a validation issue: the learning process is unlikely to encounter genuine state duplication, although detecting and resolving such duplication is often crucial for correctness and termination. For random deterministic automata, minimality and related structural properties have been studied extensively \cite{bassino2012asymptotic,babaali2013construction}. To the best of our knowledge, the analogous high-probability minimality/non-minimality question has not been studied for random Moore automata. We address this gap by establishing a high-probability minimality result for a model of random Moore automata.

In the classical setting, two states are identified if no continuation distinguishes their future behavior. For PDFAs, a state carries an output distribution \cite{vidal2005probabilistic,balle2013learning}, and more generally, for Moore automata, each state carries an observable output label \cite{PEYRIERE2023113774}. Moreover, in applications such as constrained PDFA learning and constrained or grammar-guided generation for large language models \cite{carrasco2024analyzing,willard2023efficient,raspanti2025grammar}, the continuations available for comparison may be restricted by constraints imposed by the generation procedure. Thus, two prefixes should be compared only through continuations that are admissible after both prefixes. This leads naturally to prefix-dependent congruences, which differ from the classical Myhill--Nerode right congruence.

We consider deterministic Moore automata equipped with a prefix-dependent admissibility rule. In the finite-catalog case studied below, each state receives a label from a finite catalog. The label determines both the observable output and the set of admissible non-terminal symbols. A PDFA with finitely many possible output distributions is recovered by taking the observable output to be the corresponding probability distribution.

The random model is based on a random transition system on \(n\) states, in which all transition values are independent and uniformly distributed over the state set. In this model, we prove a high-probability minimality theorem. More precisely, if \(\boldsymbol A_\infty(n)\) denotes the set of non-diagonal pairs identified by the prefix-dependent congruence, then, under mild assumptions, we have \(\mathbb P(|\boldsymbol A_\infty(n)|>0)\to0\) as \(n\to\infty\).

The classical DFA model is recovered by taking a binary catalog and declaring all symbols admissible. Hence, our proof gives an alternative argument in the iid transition model, parallel to the high-probability minimality result of \cite{bassino2012asymptotic} for random accessible DFAs, extending the random-DFA minimality picture to finite catalogs with label-dependent admissible continuations.

For automata-theoretic applications, the reachable part from the initial state is often the relevant object. Full minimality immediately implies that the restriction of the congruence to the reachable part is also trivial with high probability.

The proof is based on two ideas. First, we follow the pairs of states that remain indistinguishable after successive rounds of comparison; these rounds define a natural pruning process. A pair survives one more round only if the two states have the same catalog label and their admissible successors also survived the previous round. Up to the required number of pruning rounds, the dependencies created by this exploration are sparse enough that the early evolution is controlled by a deterministic collision-free recursion, in the spirit of the differential-equation method for random processes \cite{Wormald1999,bennett2022gentle}.

The second idea is the main technical contribution of the paper. It uses the structure forced by transitivity and stability under admissible transitions: pairs that survive forever cannot be arranged arbitrarily, but must be the pairs inside nontrivial congruence classes. We show by a first-moment argument that, with high probability, no such nontrivial family of classes exists at the scale reached by the pruning process. This turns the intermediate pruning estimate into actual minimality.

The paper is organized as follows. Section~\ref{sec:automata-congruences} recalls Moore automata, PDFAs as a motivating special case, and introduces prefix-dependent congruences. Section~\ref{sec:random-model} defines the random finite-catalog model and states the standing assumptions. Section~\ref{sec:pair-pruning} introduces the pair-pruning process and proves the collision-free estimate. Section~\ref{sec:stable-block-systems} rules out nontrivial stable block systems at the scale reached by the pruning process. Section~\ref{sec:main-results} derives minimality. Section~\ref{sec:conclusion} discusses the accessible model and possible extensions.

%% file: section2-automata-and-congruences.tex
\section{Automata and prefix-dependent congruences}
\label{sec:automata-congruences}

Let \(\Sigma\) be a finite alphabet and write \(\Sigma^\ast\) for the free monoid of finite words over \(\Sigma\), with identity element \(\varepsilon\).

A deterministic Moore automaton is a tuple
\[
\mathcal A=(Q,\Sigma,\delta,q_0,o),
\]
where \(Q\) is a finite set of states, \(q_0\in Q\) is the initial state, \(\delta:Q\times\Sigma\to Q\) is a transition map, and \(o:Q\to\mathcal O\) assigns to each state an output value in a set \(\mathcal O\); see e.g.\ \cite{hopcroft2001automata}.

We extend \(\delta\) to a map \(\delta^\ast:Q\times\Sigma^\ast\to Q\) in the usual way: \(\delta^\ast(q,\varepsilon)=q\) and \(\delta^\ast(q,aw)=\delta^\ast(\delta(q,a),w)\). We also define \(o^\ast(q,w):=o(\delta^\ast(q,w))\). We simply write \(\delta^\ast(w)\) and \(o^\ast(w)\) in the special case \(q=q_0\).

Probabilistic deterministic finite automata are a motivating special case. If \(\$\notin\Sigma\) is a terminal symbol and \(\Sigma_\$=\Sigma\cup\{\$\}\), then a PDFA is obtained by taking \(\mathcal O=\Delta(\Sigma_\$)\) the probability simplex over $\Sigma_\$$, and \(o:Q\to\Delta(\Sigma_\$)\) assigns to each state an output distribution. Thus the Moore framework keeps the observable state data abstract, while PDFAs correspond to the case in which the output value is a probability distribution.

\subsection{Prefix-dependent admissibility}

We now introduce the mechanism that determines which continuations are admissible after a given prefix. In full generality, an admissibility scheme is a map \(S:\Sigma^\ast\to 2^\Sigma\), where \(S(u)\subseteq\Sigma\) is the set of symbols admissible after the prefix \(u\).

The admissibility predicate \(\Adm:\Sigma^\ast\to\{0,1\}\) is defined recursively by \(\Adm(\varepsilon)=1\) and for \(u\in\Sigma^\ast,a\in\Sigma\) by
\begin{equation}\label{eq:AdmDef}
\Adm(ua)=1
\Leftrightarrow
\Adm(u)=1
\text{ and }
a\in S(u).
\end{equation}
We say that \(u\) is admissible if \(\Adm(u)=1\).

The observation associated with a prefix is
\[
\Obs(u):=
\bigl(o^\ast(u),S(u)\bigr).
\]
Thus \(\Obs(u)\) records both the output reached after \(u\) and the set of symbols admissible after \(u\).

\begin{definition}[Prefix-dependent congruence on words]\label{def:congruence}
For \(u,v\in\Sigma^\ast\), we write \(u\equiv v\) iff \(\Adm(u)=\Adm(v)\) and for every \(w\in\Sigma^\ast\),
\[
\Adm(uw)=\Adm(vw)=1
\Rightarrow
\Obs(uw)=\Obs(vw).
\]
\end{definition}

In words, two prefixes are equivalent if they are simultaneously admissible or inadmissible, and no common admissible continuation distinguishes the observable data reached from them. All inadmissible prefixes are identified by this relation, consistently with the absorbing property: \(\Adm(u)=0\) implies \(\Adm(uw)=0\) for every \(w\in\Sigma^\ast\).

\begin{proposition}[Right congruence]
The relation \(\equiv\) is a right congruence on \(\Sigma^\ast\). That is, for all \(u,v\in\Sigma^\ast\),
\[
u\equiv v
\Rightarrow
ua\equiv va,
\quad
\text{for every }a\in\Sigma.
\]
\end{proposition}

\begin{proof}
Let \(u\equiv v\) and \(a\in\Sigma\). First, \(\Adm(ua)=\Adm(va)\). Indeed, if \(\Adm(u)=\Adm(v)=0\), this follows from the absorbing definition of \(\Adm\). If \(\Adm(u)=\Adm(v)=1\), then applying \eqref{eq:AdmDef} and Definition \ref{def:congruence} with \(w=\varepsilon\) gives \(\Obs(u)=\Obs(v)\), hence \(S(u)=S(v)\), and therefore
\[
\Adm(ua)=1
\Leftrightarrow
a\in S(u)
\Leftrightarrow
a\in S(v)
\Leftrightarrow
\Adm(va)=1.
\]
Now let \(w\in\Sigma^\ast\) and suppose that
\[\Adm((ua)w)=\Adm((va)w)=1.\]
Since \(u\equiv v\), applying Definition \ref{def:congruence} with the continuation \(aw\) gives \(\Obs(u(aw))=\Obs(v(aw))\). Hence \(ua\equiv va\).
\end{proof}

\subsection{Guided systems and admissibility schemes}

Admissibility schemes arise naturally from guided generation in the context of language models and PDFAs. Let \(\mathcal A=(Q,\Sigma,\delta,q_0,o)\) be a PDFA and let \(\mathcal G=(R,\Sigma,\eta,r_0,\mathcal S)\) be a guide automaton where \(\mathcal S:R\to 2^\Sigma\) assigns admissible symbols to guide states.

The product system
\[
\mathcal A\otimes\mathcal G
=
(Q\times R,\Sigma,\tau,(q_0,r_0),O)
\]
is defined by \(\tau((q,r),a):=(\delta(q,a),\eta(r,a))\) and \(O(q,r):=o(q)\). It becomes an admissibility scheme by setting \(S_{(q,r)}:=\mathcal S(r)\). Thus a prefix-dependent rule driven by the guide state is represented as a state-dependent admissibility rule on the product automaton.

This covers constrained or guided generation schemes in which the output \(o(q)\) is a next-symbol distribution while the guide state restricts the symbols that may be generated next. One natural example is the top-\(k\) rule: one takes \(S_{q}\) to be the set of the \(k\) non-terminal symbols of largest probability under \(o(q)\). Another example is obtained when the guide has final states \(F\subseteq R\) and \(\mathcal S(r)\) consists of those symbols \(a\) for which \(\eta(r,a)\) is coaccessible to \(F\).

\subsection{Finite admissibility schemes}

The prefix-dependent schemes used in the sequel are finite in the following sense: admissibility is determined by the current state of the Moore automaton. Thus each state \(q\in Q\) is equipped with a set \(S_q\subseteq\Sigma\) and the induced prefix-dependent rule is \(S(u):=S_{\delta^\ast(u)}\). Equivalently, after reading a prefix \(u\), the symbols admissible for the next step are those in the set attached to the reached state \(\delta^\ast(u)\). In this case the observable attached to a state is
\[
\obs(q):=(o(q),S_q)\in\mathcal O\times 2^\Sigma,
\]
and the corresponding prefix observation satisfies
\[
\Obs(u)=\obs(\delta^\ast(u)).
\]
A word \(w=a_1\cdots a_t\) is admissible from a state \(q\) if, writing \(q_i:=\delta^\ast(q,a_1\cdots a_i)\), one has \(a_i\in S_{q_{i-1}}\) for every \(i=1,\dots,t\).

\begin{definition}[State congruence]
For \(q,q'\in Q\), we write \(q\equiv q'\) iff for every word \(w\in\Sigma^\ast\) admissible from both \(q\) and \(q'\), we have
\[
\obs(\delta^\ast(q,w))
=
\obs(\delta^\ast(q',w)).
\]
\end{definition}

Since the empty word is admissible from every state, the definition implies \(\obs(q)=\obs(q')\). In particular, because \(\obs(q)=(o(q),S_q)\), equivalent states have the same observable output and the same admissible set: \(o(q)=o(q')\) and \(S_q=S_{q'}\).

\begin{remark}[Recursive characterization]\label{rem:recursive}
The state congruence admits the following recursive characterization:
\(q\equiv q'\) iff \(\obs(q)=\obs(q')\) and
\[
\delta(q,a)\equiv \delta(q',a)
\quad
\text{for every }a\in S_q .
\]
Indeed, \(\obs(q)=\obs(q')\) implies \(S_q=S_{q'}\). Hence the one-step admissible symbols from \(q\) and \(q'\) coincide. The condition for continuations of length zero gives equality of observations, while any nonempty admissible continuation has the form \(aw\), with \(a\in S_q=S_{q'}\) and \(w\) admissible from both successor states. Thus testing all admissible continuations from \(q\) and \(q'\) is equivalent to testing all admissible continuations from the corresponding admissible successors. The claim follows by induction on the length of admissible continuations.
\end{remark}

\begin{definition}[Minimality]
The automaton \(\mathcal A\) is minimal with respect to \(\equiv\) if every equivalence class on the reachable part of \(Q\) is a singleton.
\end{definition}

The relation between the prefix-level congruence and the state congruence has to be interpreted with the absorbing inadmissible class in mind. If \(u\) and \(v\) are admissible prefixes, then the future tests after \(u\) and \(v\) are exactly the admissible tests from the reached states \(\delta^\ast(u)\) and \(\delta^\ast(v)\). Hence \(u\equiv v\)
iff \(\delta^\ast(u)\equiv\delta^\ast(v)\) whenever \(\Adm(u)=\Adm(v)=1\). On the other hand, if \(\Adm(u)=\Adm(v)=0\), then \(u\equiv v\) by the absorbing property, independently of the states reached by the complete transition map. Thus the full prefix congruence can be written as
\[
u\equiv v
\quad\Leftrightarrow\,
\begin{array}{l}
\Adm(u)=\Adm(v) \text{ and}\\
\Adm(u)=0
\text{ or }
\delta^\ast(u)\equiv\delta^\ast(v).
\end{array}
\]
Thus, the word-level quotient consists of the state quotient together with one absorbing inadmissible class.

Equivalently, one may adjoin a sink state \(\bot\notin Q\) and extend the transition map by sending every inadmissible prefix to \(\bot\). In this formulation, \(\bot\) represents the single absorbing class of inadmissible prefixes. We will not use this sink-state formulation. Thus inadmissible prefixes form a single absorbing class at the word level, whereas all state-level congruences are considered on the original state space \(Q\).

%% file: section3-random-model.tex
\section{The random model}
\label{sec:random-model}

We now define the random finite-catalog model. The alphabet \(\Sigma\) is fixed throughout and \(m:=|\Sigma|\).

\subsection{Transition structure}

For \(n\geq1\), let \(Q_n=[n]=\{1,\dots,n\}\). Let \(\operatorname{Tra}_{n,m}\) denote the set of all complete deterministic transition structures on \(Q_n\) over \(\Sigma\). Thus an element of \(\operatorname{Tra}_{n,m}\) is a map
\[
\delta:Q_n\times\Sigma\to Q_n.
\]
Throughout the paper we use the random transition model
\[
\boldsymbol\delta\sim\operatorname{Unif}(\operatorname{Tra}_{n,m}).
\]
Equivalently, the random variables \(\boldsymbol\delta(q,a)\), \(q\in Q_n,\ a\in\Sigma\), are independent and uniformly distributed on \(Q_n\). We use bold font to denote random objects.

The distinguished initial state is \(1\). The reachable part of a transition structure is \(\operatorname{acc}(\delta) := \{\delta^\ast(1,u):u\in\Sigma^\ast\}\). In this model the reachable part is typically a proper subset of \(Q_n\), but it has linear size with high probability \cite{carayol2012distribution}.

\subsection{Finite-catalog model}

Let \(\mathcal D=[D]\). For each \(j\in\mathcal D\), fix an output value \(o_j\in\mathcal O\) and an admissible set \(S_j\subseteq\Sigma\). We assume that the observable catalog entries are distinct:
\[
(o_i,S_i)\neq(o_j,S_j)
\quad\text{whenever } i\neq j.
\]
Thus equality of observable data is equivalent to equality of catalog labels.

\begin{definition}[Finite-catalog random Moore automaton]
Let \(\rho=(\rho_1,\dots,\rho_D)\in\Delta(\mathcal D)\) be a probability
distribution with \(\rho_j>0\) for all \(j\). Independently of
\(\boldsymbol\delta\), sample iid labels
\[
\boldsymbol\xi(q)\sim\rho,
\quad q\in Q_n.
\]
The resulting random output automaton is
\[
\boldsymbol{\mathcal A}_n=(Q_n,\Sigma,\boldsymbol\delta,1,\boldsymbol o),
\]
where \(\boldsymbol o(q):=o_{\boldsymbol\xi(q)}\). The admissible set at \(q\) is \(\boldsymbol S_q:=S_{\boldsymbol\xi(q)}\).
\end{definition}

For each catalog label \(j\), define
\[
k_j:=|S_j|,
\qquad
k_\ast:=\min_{1\leq j\leq D}k_j.
\]
The main result is proved under the following uniform lower bound on the number of admissible symbols:
\begin{equation}\label{eq:kstar-assumption}
k_\ast\geq 3.
\end{equation}
We also assume that the label law is non-degenerate:
\begin{equation}\label{eq:nondegenerate-rho}
    s_0:=\sum_{j=1}^D\rho_j^2<1.
\end{equation}
Equivalently, two independently sampled labels do not agree with probability
one. Since all \(\rho_j\) are assumed positive, this is the same as requiring
\(D\geq2\).

\subsection{PDFA top-\(k\) example}

A simple PDFA example is obtained by choosing a finite set of distinct output distributions \(\mathcal P=\{p_1,\dots,p_D\}\subseteq\Delta(\Sigma_\$)\), taking \(\mathcal O=\Delta(\Sigma_\$)\) and \(o_j=p_j\), and setting \(S_j\) to be the set of the \(k\) largest-probability non-terminal symbols under \(p_j\). Then \(k_j=k\) for all \(j\), so \(k_\ast=k\). If the labels are sampled uniformly from the catalog, then
\(\rho_j=1/D\), \(j=1,\dots,D\), and therefore
\[
s_0=\sum_{j=1}^D\rho_j^2=\frac1D.
\]

%% file: section4-pair-pruning.tex
\section{Pair pruning and collision-free evolution}
\label{sec:pair-pruning}

Let \(\mathcal E_n:=\bigl\{\{q,q'\}:q,q'\in Q_n,\ q\neq q'\bigr\}\) be the set of unordered non-diagonal pairs of states. We use the notation \(\operatorname{Diag}_n:=\{\{q\}:q\in Q_n\}\). If \(e=\{q,q'\}\in\mathcal E_n\) and \(a\in\Sigma\), write \(e\cdot a:=\{\boldsymbol{\delta}(q,a),\boldsymbol{\delta}(q',a)\}\in \mathcal E_n\cup \operatorname{Diag}_n\). That is, \(e\cdot a\) is diagonal if \(\boldsymbol{\delta}(q,a)=\boldsymbol{\delta}(q',a)\).

\subsection{The pruning process}

Define \(\boldsymbol A_0:=\{\{q,q'\}\in\mathcal E_n:\boldsymbol{\xi}(q)=\boldsymbol{\xi}(q')\}\) and for \(t\geq0\) set
\[
\begin{aligned}
\boldsymbol A_{t+1}:=
\bigl\{
e\in \boldsymbol A_0:
e\cdot a\in \boldsymbol A_t\cup \operatorname{Diag}_n
\, \forall a\in \boldsymbol S_e
\bigr\}
\end{aligned}
\]
where \(\boldsymbol S_e:=S_j\) if \(e=\{q,q'\}\) and \(\boldsymbol{\xi}(q)=\boldsymbol{\xi}(q')=j\). The sequence is decreasing: \(\boldsymbol A_t\supseteq \boldsymbol A_{t+1}\) for all \(t\geq 0\). We define
\[
\boldsymbol A_\infty(n):=\bigcap_{t\geq0}\boldsymbol A_t.
\]
The set \(\boldsymbol A_\infty(n)\) consists exactly of those unordered pairs whose endpoints belong to the same equivalence class of the state congruence defined by the finite admissibility scheme:
\begin{equation}\label{eq:Ainfty-equivalence}
\boldsymbol A_\infty(n)=\{\{q,q'\}\in\mathcal{E}_n:q\boldsymbol{\equiv} q'\}.
\end{equation}
Thus, the relation defined by \(q=q'\) or \(\{q,q'\}\in \boldsymbol A_\infty(n)\) coincides with \(\boldsymbol{\equiv}\).

\subsection{Collision-free recursion}

For \(j\in[D]\), let \(y_{t,j}\) denote the approximation obtained by ignoring
collisions, namely the probability that a fixed pair survives \(t\) rounds and
has common label \(j\) in the collision-free evolution. Define
\[
y_{0,j}:=\rho_j^2,
\quad
s_t:=\sum_{j=1}^D y_{t,j},
\]
and recursively
\begin{equation}\label{eq:mf-recursion}
y_{t+1,j}:=\rho_j^2 s_t^{k_j}.
\end{equation}
We refer to \eqref{eq:mf-recursion} as the collision-free recursion; the term \(s_t^{k_j}\) reflects the approximation that the \(k_j\) admissible successor comparisons behave independently. Thus
\begin{equation}\label{eq:s-recursion}
    s_{t+1}=\sum_{j=1}^D\rho_j^2s_t^{k_j}.
\end{equation}
Because \(k_j\geq k_\ast\) and \(s_t\leq s_0<1\),
\begin{equation}\label{eq:fast-decay}
    s_{t+1}\leq s_t^{k_\ast},
    \qquad
    s_t\leq s_0^{k_\ast^t}.
\end{equation}

\subsection{Local exploration}

Fix \(e=\{q,q'\}\in\mathcal E_n\). The depth-\(t\) exploration of \(e\) reveals the states \(\boldsymbol\delta^\ast(q,u)\) and \(\boldsymbol\delta^\ast(q',u)\) for all words \(u\) that may appear along an admissible branch of length at most \(t\).

Since every admissible set has size at most \(m\), the number of pair-nodes in this exploration is bounded by
\[
B_t:=1+m+\cdots+m^t
=
\frac{m^{t+1}-1}{m-1}.
\]
Thus, the exploration reveals at most \(2B_t\) state occurrences.

Let \(\boldsymbol{W}_t(e)\) be the set of words of length at most \(t\) that are explored from the pair \(e=\{q,q'\}\); that is, at each intermediate pair reached along the word, the next symbol is admissible for the common catalog label of that pair.  We define the collision event \(C_t(e)\) as the event that the map
\[
\{q,q'\}\times \boldsymbol W_t(e)\to Q_n,
\quad
(\eta,u)\mapsto \boldsymbol\delta^\ast(\eta,u)
\]
is not injective. Equivalently, two distinct state occurrences in the depth-\(t\) exploration of \(e\) correspond to the same state of \(Q_n\).

\begin{lemma}[Collision bound for local exploration]\label{lem:tree-like}
For every fixed pair \(e\in\mathcal E_n\) and every \(t\geq0\), we have
\[
\mathbb P\left(C_t(e)\right)
\leq
\frac{2B_t^2}{n}.
\]
\end{lemma}

\begin{proof}
Fix \(e=\{q,q'\}\). We expose the depth-\(t\) exploration of \(e\) in a deterministic breadth-first order, breaking ties by a fixed ordering of \(\Sigma\). Let \(\boldsymbol q_1,\boldsymbol q_2,\ldots,\boldsymbol q_{\boldsymbol N}\) be the sequence of state occurrences revealed in this order. The number \(\boldsymbol N\) may depend on the admissible branches encountered during the exploration, but always satisfies \(\boldsymbol N\leq 2B_t\).

Let \(\mathcal F_i:=\sigma(\boldsymbol q_1,\ldots,\boldsymbol q_i)\) be the filtration associated with the exploration process and define the first collision time
\[
\boldsymbol\tau:=\inf\{i\geq 2: \boldsymbol q_i\in\{\boldsymbol q_1,\ldots,\boldsymbol q_{i-1}\}\},
\]
with the convention \(\boldsymbol\tau=\infty\) if no such \(i\) exists. Then
\[
C_t(e)=\{\boldsymbol\tau\leq \boldsymbol N\}.
\]

For \(2\leq i\leq 2B_t\), on the event \(\{i\leq \boldsymbol\tau, \boldsymbol N\}\), no collision has occurred before \(\boldsymbol q_i\) is exposed. Hence \(\boldsymbol q_1,\ldots,\boldsymbol q_{i-1}\) are distinct. Apart from the deterministic root occurrences \(q\) and \(q'\), each new state occurrence \(\boldsymbol q_i\) is obtained by revealing a transition value which, conditionally on \(\mathcal F_{i-1}\), is uniform on \(Q_n\) and independent of the previously revealed transition values. For the deterministic root occurrences, the same bound below is trivial. Therefore, denoting
\[
E_i:=
\left\{
\boldsymbol q_i\in\{\boldsymbol q_1,\ldots,\boldsymbol q_{i-1}\},
\, i\leq \boldsymbol N
\right\},
\]
we have
\[
\begin{aligned}
\mathbb P\left(
E_i
\,\middle|\,
\mathcal F_{i-1}
\right)
\mathds{1}_{\{\boldsymbol \tau\geq i\}}
=
\frac{i-1}{n}\,
\mathds{1}_{\{\boldsymbol \tau\geq i\}}.
\end{aligned}
\]
Since
\[
\{\boldsymbol\tau=i\leq \boldsymbol N\}
\subseteq
\{\boldsymbol\tau\geq i\}
\cap
E_i,
\]
we have
\[
\begin{aligned}
\mathbb P(\boldsymbol \tau=i\leq \boldsymbol N)
&\leq
\mathbb E\left[
\mathds{1}_{\{\boldsymbol \tau\geq i\}}
\mathbb P\left(
E_i
\,\middle|\,
\mathcal F_{i-1}
\right)
\right]  \\
&=
\mathbb E\left[
\mathds{1}_{\{\boldsymbol \tau\geq i\}}
\frac{i-1}{n}
\right]
\leq
\frac{i-1}{n}.
\end{aligned}
\]
Since \(\boldsymbol N\leq 2B_t\), we get
\[
\begin{aligned}
\mathbb P(C_t(e))
& =
\mathbb P(\boldsymbol \tau\leq \boldsymbol N)
\leq
\sum_{i=2}^{2B_t}\mathbb P(\boldsymbol \tau=i\leq \boldsymbol N)\\
&\leq
\sum_{i=2}^{2B_t}\frac{i-1}{n}
=
\frac{2B_t(2B_t-1)}{2n}
\leq
\frac{2B_t^2}{n}.
\end{aligned}
\]
This completes the proof.
\end{proof}

\begin{proposition}[Collision-free approximation for one pair]\label{prop:one-pair-mf}
For every \(t\geq0\), every \(j\in\mathcal{D}\), and every \(e=\{q,q'\}\in\mathcal E_n\), we have
\[
\mathbb P\left(e\in \boldsymbol A_t,\boldsymbol \xi(q)=\boldsymbol \xi(q')=j\right)
=
y_{t,j}+O\left(\frac{B_t^2}{n}\right).
\]
Consequently
\(
\mathbb P(e\in \boldsymbol A_t)
=
s_t+O\left(\frac{B_t^2}{n}\right).
\)
\end{proposition}

\begin{proof}
Fix \(e=\{q,q'\}\in\mathcal E_n\), \(j\in\mathcal{D}\), and \(t\geq0\). Write
\[
H_{t,j}(e)
:=
\{e\in \boldsymbol A_t,\boldsymbol \xi(q)=\boldsymbol \xi(q')=j\}.
\]
We compare the actual depth-\(t\) exploration of \(e\) with an ideal collision-free tree process.

Define the auxiliary process as follows. To each word \(u\in \Sigma^\ast\) we attach an independent pair of catalog labels
\[
(\boldsymbol\zeta_u,\boldsymbol\zeta'_u)\in\mathcal D^2,
\]
with common law \(\rho\otimes\rho\). The root is the empty word \(\varepsilon\). For \(r\geq0\), define recursively the events \(H^{\text{CF}}_{r,j}(u)\) where \(u\in\Sigma^\ast\) and \(j\in\mathcal{D}\) by
\[
H^{\text{CF}}_{0,j}(u)
:=
\{\boldsymbol\zeta_u=\boldsymbol\zeta'_u=j\},
\]
and for \(r\geq0\)
\[
H^{\text{CF}}_{r+1,j}(u)
:=
\{\boldsymbol\zeta_u=\boldsymbol\zeta'_u=j\}
\cap
\bigcap_{a\in S_j}
\bigcup_{\ell=1}^D
H^{\text{CF}}_{r,\ell}(ua).
\]
In words, this is the event that the two labels at \(u\) are both equal to \(j\), and that for every admissible symbol \(a\in S_j\) the child \(ua\) satisfies the depth-\(r\) survival condition with some common label \(\ell\in\mathcal{D}\).

We first compute the auxiliary survival probabilities. Clearly,
\[
\mathbb P\left(H^{\text{CF}}_{0,j}(u)\right)
=
\mathbb P(\boldsymbol\zeta_u=\boldsymbol\zeta'_u=j)
=
\rho_j^2
=
y_{0,j}.
\]
Suppose that for some \(r\geq0\)
\[
\mathbb P\left(H^{\text{CF}}_{r,\ell}(u)\right)=y_{r,\ell}
\qquad
\text{for every } \ell\in\mathcal{D}
\]
and every \(u\in\Sigma^\ast\). Then, for each admissible child \(ua\)
\[
\mathbb P\left(
\bigcup_{\ell=1}^D H^{\text{CF}}_{r,\ell}(ua)
\right)
=
\sum_{\ell=1}^D y_{r,\ell}
=
s_r,
\]
because the events \(H^{\text{CF}}_{r,\ell}(ua)\), \(\ell\in\mathcal{D}\), are disjoint. Moreover, the subtrees rooted at the children \(ua\), \(a\in S_j\), are
independent. Hence
\[
\mathbb P\left(H^{\text{CF}}_{r+1,j}(u)\right)
=
\rho_j^2 s_r^{k_j}
=
y_{r+1,j},
\]
where \(k_j=|S_j|\). By induction, \(\mathbb P\left(H^{\text{CF}}_{t,j}(\varepsilon)\right)=y_{t,j}\).

It remains to compare the true exploration with the auxiliary collision-free process. On \(C_t(e)^c\) the map
\[
\{q,q'\}\times \boldsymbol W_t(e)\to Q_n,
\qquad
(\eta,u)\mapsto \boldsymbol\delta^\ast(\eta,u),
\]
is injective. Therefore all state occurrences revealed by the depth-\(t\) exploration are distinct. Since the random labels \(\{\boldsymbol\xi(x):x\in Q_n\}\) are independent with law \(\rho\), the labels seen at the distinct state occurrences of the exploration are independent with law \(\rho\). Moreover, before any collision occurs different successor branches are generated by revealing fresh transition values, and hence the corresponding descendant explorations are independent.

Consequently, the true exploration and the auxiliary process may be coupled so that on \(C_t(e)^c\) we have \(\mathds{1}_{H_{t,j}(e)}=\mathds{1}_{H^{\text{CF}}_{t,j}(\varepsilon)}\). It follows that
\[
\begin{aligned}
\left|
\mathbb P\left(H_{t,j}(e)\right)
-
\mathbb P\left(H^{\text{CF}}_{t,j}(\varepsilon)\right)
\right|
&\leq
\mathbb P\left(
\mathds{1}_{H_{t,j}(e)}
\neq
\mathds{1}_{H^{\text{CF}}_{t,j}(\varepsilon)}
\right)  \\
&\leq
\mathbb P(C_t(e)).
\end{aligned}
\]
By Lemma~\ref{lem:tree-like}, the last term is upper-bounded by \(\frac{2B_t^2}{n}\). Therefore
\(
\left|
\mathbb P\bigl(e\in \boldsymbol A_t,\ \boldsymbol\xi(q)=\boldsymbol\xi(q')=j\bigr)
-
y_{t,j}
\right|
\leq
\frac{2B_t^2}{n}.
\)
Finally, define
\(
H_t(e):=\{e\in \boldsymbol A_t\}\),
\(H_t^{\mathrm{CF}}(\varepsilon)
:=
\bigcup_{j=1}^D H^{\mathrm{CF}}_{t,j}(\varepsilon).
\)
The events \(H^{\mathrm{CF}}_{t,j}(\varepsilon)\), \(j\in\mathcal D\), are pairwise disjoint. Therefore
\[
\mathbb P\left(H_t^{\mathrm{CF}}(\varepsilon)\right)
=
\sum_{j=1}^D y_{t,j}
=
s_t.
\]
Under the same coupling, on \(C_t(e)^c\) we also have
\(
\mathds{1}_{H_t(e)}
=
\mathds{1}_{H_t^{\mathrm{CF}}(\varepsilon)}.
\)
Hence
\(
\left|
\mathbb P(e\in \boldsymbol A_t)-s_t
\right|
\leq
\mathbb P(C_t(e))
\leq
\frac{2B_t^2}{n}.
\)
This completes the proof.
\end{proof}

\begin{corollary}[Reduction to an intermediate scale]\label{cor:mesoscopic-entry}
Assume that \(s_0<1\) as in \eqref{eq:nondegenerate-rho}. For every \(0<\beta<1\), there exists a deterministic number of rounds
\[
t_n=O_{\beta,k^\ast,s_0}(\log\log n)
\]
such that
\(
\mathbb P\left(|\boldsymbol A_{t_n}|>n^{1+\beta}\right)\to0
\)
when \(n\to\infty\).
\end{corollary}

\begin{proof}
Fix \(0<\beta<1\), and choose an auxiliary exponent \(\alpha:=\beta/2\). Define
\[
t_n:=\left\lceil\log_{k_\ast}\left(\frac{(1-\alpha)\log n}{-\log s_0}\right)\right\rceil.
\]
Then \(t_n=O(\log\log n)\). Moreover, \(s_{t_n}\leq n^{-(1-\alpha)}\) by \eqref{eq:fast-decay}. Since \(B_t\leq2 m^t\) and \(t_n=O(\log\log n)\), there exists a constant \(C>0\), depending on \(m\), such that \(B_{t_n}\leq (\log n)^C\). Hence
\[
\frac{B_{t_n}^2}{n}
\leq
\frac{(\log n)^{2C}}{n}.
\]
By Proposition~\ref{prop:one-pair-mf}, uniformly over fixed \(e\in\mathcal E_n\)
\[
\mathbb P(e\in \boldsymbol A_{t_n})
=
s_{t_n}
+
O\left(\frac{B_{t_n}^2}{n}\right).
\]
Therefore, by linearity of expectation
\[
\begin{aligned}
\mathbb E|\boldsymbol A_{t_n}|
&=
\sum_{e\in\mathcal E_n}
\mathbb P(e\in \boldsymbol A_{t_n}) \\
&\leq
\binom n2
\left(
n^{-(1-\alpha)}
+
C_1\frac{(\log n)^{2C}}{n}
\right) \\
&\leq
C_2 n^{1+\alpha}
+
C_3 n(\log n)^{2C}.
\end{aligned}
\]
Since \(\alpha>0\), \(n(\log n)^{2C}=o(n^{1+\alpha})\). Thus, for all sufficiently large \(n\)
\[
\mathbb E|\boldsymbol A_{t_n}|
\leq
C_4 n^{1+\alpha}.
\]
Finally, Markov's inequality gives
\[
\begin{aligned}
\mathbb P\left(|\boldsymbol A_{t_n}|>n^{1+\beta}\right)
&\leq
\frac{\mathbb E|\boldsymbol A_{t_n}|}{n^{1+\beta}}
&\leq
C_4 n^{-\beta/2}\to 0
\end{aligned}
\]
as claimed.
\end{proof}

%% file: section5-stable-block-systems.tex
\section{Stable block systems and obstruction bounds}
\label{sec:stable-block-systems}

Corollary~\ref{cor:mesoscopic-entry} shows that, for every \(0<\beta<1\), the pruning process reaches size at most \(n^{1+\beta}\) after \(O(\log\log n)\) steps. To upgrade this estimate to minimality, we use the fact that the limiting set is generated by equivalence classes of the state congruence. Such classes cannot be arbitrary: they must have constant observable catalog label and be stable under admissible transitions. We formalize this structure through stable block systems and prove a first-moment bound showing that, with high probability, no nonempty stable block system of size at most \(n^{1+\beta}\) exists, provided \(k_\ast\geq3\) and \(\beta<1-2/k_\ast\).

\subsection{Stable block systems}

Let \(\mathcal B=\{B_1,\dots,B_L\}\) be a family of pairwise disjoint subsets of \(Q_n\), with \(|B_i|=b_i\geq2\). The family \(\mathcal B\) defines a partial partition of \(Q_n\): the blocks \(B_i\) are the non-singleton classes and every point outside
\[
V(\mathcal B):=\bigcup_{i=1}^L B_i
\]
is a singleton class. Write \(\sim_{\mathcal B}\) for this equivalence relation.

Define
\[
r(\mathcal B):=\sum_{i=1}^L b_i
\text{ and }
h(\mathcal B):=\sum_{i=1}^L \binom{b_i}{2}.
\]
Thus \(h(\mathcal B)\) is the number of unordered non-diagonal pairs identified by \(\sim_{\mathcal B}\).

\begin{definition}[Stable block system]
A block system \(\mathcal B\) is called stable if the following conditions hold.
\begin{enumerate}[label=(\roman*)]
\item Each block has a constant catalog label: for each \(i\), there is a label \(j_i\in\mathcal{D}\) such that \(\boldsymbol\xi(q)=j_i\) for all \(q\in B_i\).
\item For every block \(B_i\), admissible symbol \(a\in S_{j_i}\) and \(q,q'\in B_i\), \(\boldsymbol\delta(q,a)\sim_{\mathcal B}\boldsymbol\delta(q',a)\).
\end{enumerate}
Equivalently, for each \(i\) and \(a\in S_{j_i}\) the image set
\[
\boldsymbol\delta(B_i,a)
:=
\{\boldsymbol\delta(q,a):q\in B_i\}
\]
is contained in a single class of \(\sim_{\mathcal B}\), where singleton classes are included.
\end{definition}

\begin{lemma}\label{lem:Ainfty-gives-stable-block-system}
If \(|\boldsymbol A_\infty(n)|>0\), then the non-singleton equivalence classes of \(\boldsymbol{\equiv}\) form a stable block system \(\boldsymbol{\mathcal B}_\infty\) and
\[
h(\boldsymbol{\mathcal B}_\infty)=|\boldsymbol A_\infty(n)|.
\]
\end{lemma}

\begin{proof}
By \eqref{eq:Ainfty-equivalence} \(\boldsymbol A_\infty(n)\) is exactly the set of non-diagonal pairs contained in the non-singleton classes of \(\boldsymbol{\equiv}\). Equality of state observations forces every non-singleton class to have a constant catalog label. Moreover, the congruence property of \(\boldsymbol{\equiv}\) implies stability under admissible transitions. Hence the successor pair is either diagonal or belongs to a non-singleton class of \(\boldsymbol{\equiv}\). Thus the non-singleton classes form a stable block system. The identity for \(h\) is immediate from its definition.
\end{proof}

\subsection{Probability of a fixed stable block system}

For a block system \(\mathcal B\) with block sizes \(b_1,\dots,b_L\) define
\begin{equation}\label{eq:Gamma-b}
\Gamma_b(\mathcal B)
:=
(n-r(\mathcal B))\left(\frac1n\right)^b
+
\sum_{\ell=1}^L\left(\frac{b_\ell}{n}\right)^b.
\end{equation}
This is the probability that \(X_1,\ldots,X_b\) iid uniform samples from \(Q_n\) all fall in one class of \(\sim_{\mathcal B}\).

\begin{lemma}[Probability of a fixed stable block system]
\label{lem:fixed-stable-block-system-probability}
For a deterministic block system \(\mathcal B=\{B_1,\dots,B_L\}\), let
\(E_{\mathcal B}\) be the event that \(\mathcal B\) is stable. Then
\[
\mathbb P(E_{\mathcal B})
=
\prod_{i=1}^L
\left(
\sum_{j=1}^D
\rho_j^{b_i}\,
\Gamma_{b_i}(\mathcal B)^{k_j}
\right).
\]
In particular \(\mathbb P(E_{\mathcal B})\leq \prod_{i=1}^L \Gamma_{b_i}(\mathcal B)^{k_\ast}\).
\end{lemma}

\begin{proof}
Let \(B_i\) be a block of size \(b_i\). For \(j\in\mathcal{D}\), let \(E_{i,j}\) be the event that \(B_i\) has constant catalog label \(j\). Then
\[
\mathbb P(E_{i,j})=\rho_j^{b_i}.
\]
Conditional on \(E_{i,j}\), stability of \(B_i\) requires that for every \(a\in S_j\) the image set
\[
\boldsymbol\delta(B_i,a):=\{\boldsymbol\delta(q,a):q\in B_i\}
\]
is contained in one class of \(\sim_{\mathcal B}\). For fixed \(a\in S_j\), the random variables \(\{\boldsymbol\delta(q,a):q\in B_i\}\) are iid uniform on \(Q_n\). Hence this probability is \(\Gamma_{b_i}(\mathcal B)\). Since distinct symbols use disjoint transition variables
\[
\mathbb P(B_i\text{ is stable}\mid E_{i,j})
=
\Gamma_{b_i}(\mathcal B)^{k_j}.
\]
Summing over the possible common labels gives the factor for \(B_i\). The events associated with distinct blocks depend on disjoint label variables and disjoint source transition variables, so the factors multiply.

Finally, since \(k_j\geq k_\ast\) and \(0\leq\Gamma_{b_i}(\mathcal B)\leq1\) we have
\[
\Gamma_{b_i}(\mathcal B)^{k_j}
\leq
\Gamma_{b_i}(\mathcal B)^{k_\ast}.
\]
Also \(\sum_j\rho_j^{b_i}\leq1\). This gives the stated upper bound.
\end{proof}

\subsection{Small stable block system bound}

The following estimate excludes all nonempty stable block systems up to the
scale reached by the pruning process.

\begin{lemma}[Small stable block system bound]\label{lem:small-obstruction-bound}
Assume \(k_\ast\geq3\) as in \eqref{eq:kstar-assumption} and let \(0<\beta<1-\frac{2}{k_\ast}\). Let \(E_n(\beta)\) be the event that there exists a stable block system \(\mathcal B\) with \(1\leq h(\mathcal B)\leq n^{1+\beta}\). Then \(\mathbb P(E_n(\beta))\to0\) as \(n\to \infty\).
\end{lemma}

\begin{proof}
We use a first-moment argument. Fix block sizes \(b_1,\dots,b_L\geq2\) and put
\[
r=\sum_{i=1}^L b_i,
\quad
h=\sum_{i=1}^L\binom{b_i}{2}.
\]
The number of block systems with these sizes is at most
\begin{equation}\label{eq:block-count}
\frac{1}{L!}\prod_{i=1}^L\binom{n-(b_1+\cdots+b_{i-1})}{b_i}\leq\frac{n^r}{L!\prod_{i=1}^L b_i!}.
\end{equation}
By Lemma~\ref{lem:fixed-stable-block-system-probability}, the expected
number of stable block systems with these block sizes is at most
\begin{equation}\label{eq:first-moment-sizes}
\frac{n^r}{L!\prod_{i=1}^L b_i!}
\prod_{i=1}^L \Gamma_{b_i}(\mathcal B)^{k_\ast}.
\end{equation}
We sum this bound over all size vectors satisfying \(h\leq n^{1+\beta}\).

Let \(b_{\max}:=\max_{1\leq i\leq L} b_i\). Since
\begin{equation}\label{eq:sum-b2-bound}
\sum_{i=1}^L b_i^2
=
r+2h
\leq
n+2n^{1+\beta}
\leq 3n^{1+\beta},
\end{equation}
we have
\begin{equation}\label{eq:bmax-bound}
    b_{\max}\leq \left(\sum_{i=1}^L b_i\right)^{1/2}\leq 3^{1/2} n^{(1+\beta)/2}.
\end{equation}

We now bound \(\Gamma_b(\mathcal B)\) uniformly over all such block systems.
From \eqref{eq:Gamma-b}
\[
\Gamma_b(\mathcal B)
\leq
n^{1-b}
+
\sum_{\ell=1}^L\left(\frac{b_\ell}{n}\right)^b.
\]
For \(b\geq2\) we have \(b_\ell^b=b_\ell^2b_\ell^{b-2}\leq b_\ell^2b_{\max}^{b-2}\). Therefore, by \eqref{eq:sum-b2-bound} and \eqref{eq:bmax-bound} we have
\[
\begin{aligned}
\sum_{\ell=1}^L\left(\frac{b_\ell}{n}\right)^b
&\leq
\frac{b_{\max}^{b-2}}{n^b}
\sum_{\ell=1}^L b_\ell^2
&\leq
3^{b/2} n^{(b-2)(1+\beta)/2}n^{1+\beta-b}.
\end{aligned}
\]
Since
\[
(b-2)\frac{1+\beta}{2}+1+\beta-b
=
-b\frac{1-\beta}{2},
\]
we obtain
\begin{equation}\label{eq:Gamma-uniform}
\Gamma_b(\mathcal B)
\leq
n^{1-b}
+
3^{b/2} n^{-b(1-\beta)/2},
\quad b\geq2.
\end{equation}
Define
\[
a_b
:=
\frac{n^b}{b!}
\left(
n^{1-b}
+
3^{b/2} n^{-b(1-\beta)/2}
\right)^{k_\ast},
\quad b\geq2.
\]
By \eqref{eq:Gamma-uniform} we have \(\frac{n^{b_i}}{b_i!}\Gamma_{b_i}(\mathcal B)^{k_\ast}\leq a_{b_i}\) for every block \(i\). Therefore, for fixed sizes \(b_1,\ldots,b_L\) we obtain
\[
\begin{aligned}
\frac{n^r}{L!\prod_{i=1}^L b_i!}
\prod_{i=1}^L \Gamma_{b_i}(\mathcal B)^{k_\ast}
&=
\frac1{L!}
\prod_{i=1}^L
\frac{n^{b_i}}{b_i!}\Gamma_{b_i}(\mathcal B)^{k_\ast} \\
&\leq
\frac1{L!}\prod_{i=1}^L a_{b_i}.
\end{aligned}
\]
Dropping the restriction \(h\leq n^{1+\beta}\) only increases the sum over
size vectors. Hence, for fixed \(L\) we get
\[
\begin{aligned}
&\sum_{\substack{b_1,\ldots,b_L\geq2:\ h\leq n^{1+\beta}}}
\frac{n^r}{L!\prod_{i=1}^L b_i!}
\prod_{i=1}^L \Gamma_{b_i}(\mathcal B)^{k_\ast} \\
&\qquad\leq
\frac1{L!}
\sum_{b_1,\ldots,b_L\geq2}\prod_{i=1}^L a_{b_i}
=
\frac1{L!}\left(\sum_{b=2}^n a_b\right)^L.
\end{aligned}
\]
Since \(h\geq1\) is equivalent to having \(L\geq1\), we sum over \(L\geq1\). The total contribution of all nonempty block systems with \(1\leq h\leq n^{1+\beta}\) is therefore bounded by
\[
\sum_{L\geq1}\frac1{L!}\left(\sum_{b=2}^n a_b\right)^L
=
\exp\left\{\sum_{b=2}^n a_b\right\}-1.
\]
It remains to show that \(\sum_{b=2}^n a_b\to0\).

Set \(x:=n^{1-b}\) and \(y:=C^b n^{-b(1-\beta)/2}\) where $C=3^{1/2}$. Then \(a_b=n^b(x+y)^{k_\ast}/b!\). Using
\[
(x+y)^{k_\ast}\leq 2^{k_\ast-1}(x^{k_\ast}+y^{k_\ast}),
\]
and increasing \(C\) if necessary, we get
\[
\begin{aligned}
a_b
&\leq
\frac{C^b}{b!}
\left[
n^b n^{k_\ast(1-b)}
+
n^b n^{-k_\ast b(1-\beta)/2}
\right] \\
&=
\frac{C^b}{b!}
\left[
n^{k_\ast-(k_\ast-1)b}
+
n^{-b\left(k_\ast(1-\beta)/2-1\right)}
\right].
\end{aligned}
\]
Because \(b\geq2\)
\[
k_\ast-(k_\ast-1)b
\leq
-\frac{k_\ast-2}{2}b.
\]
Also
\[
\beta<1-\frac{2}{k_\ast}
\Rightarrow
\frac{k_\ast(1-\beta)}2-1>0.
\]
Thus, with
\[
\delta:=
\min\left\{
\frac{k_\ast-2}{2},
\frac{k_\ast(1-\beta)}2-1
\right\}>0,
\]
we have \(a_b\leq (Cn^{-\delta})^b/b!\) for every \(b\geq2\). Consequently
\[
\sum_{b=2}^n a_b
\leq
\sum_{b=2}^{\infty}\frac{(Cn^{-\delta})^b}{b!}
\leq
\exp\{Cn^{-\delta}\}-1
\to0.
\]
Therefore the expected number of stable block systems satisfying
\(1\leq h(\mathcal B)\leq n^{1+\beta}\) tends to zero. Markov's inequality
gives \(\mathbb P(E_n(\beta))\to0\).
\end{proof}

\begin{remark}[Non-minimality when \(k_\ast= 1\)]
The assumption \(k_\ast\ge 3\) separates the regime treated in this paper from the cases \(k_\ast=1\) and \(k_\ast=2\) which exhibit different behavior. If \(k_\ast=1\), one expects non-minimality with high probability in general. Indeed, suppose that a catalog type \(j\) with \(\rho_j>0\) has a unique admissible symbol \(a\). Then any pair of states with catalog label \(j\) and the same \(a\)-successor coalesces after the only admissible transition and is therefore indistinguishable on all admissible continuations. In the iid full transition model, the expected number of such pairs is \(\rho_j^2(n-1)/2\), so this obstruction occurs with linear expectation.
\end{remark}

\begin{remark}[The case \(k_\ast=2\) with catalog diversity]
The proof above fails for \(k_\ast=2\) because blocks of size \(2\) become critical. However, the same argument goes through if one keeps the label factor in Lemma~\ref{lem:fixed-stable-block-system-probability}. Let
\[
\lambda_b:=\sum_{j=1}^D \rho_j^b .
\]
Assume that there exist constants \(\chi>0\) and \(C_\rho>0\) such that
\[
\lambda_b\leq C_\rho^b n^{-\chi(b-1)},
\quad b\geq2.
\]
For instance, this holds for uniform labels on \(D_n=n^\chi\) types. Then, if \(k_\ast=2\) and \(\chi>2\beta\), the same first-moment proof gives
\[
\mathbb P\left(
\begin{array}{c}
\exists\text{ stable block system }\mathcal B\text{ with} \\
1\leq h(\mathcal B)\leq n^{1+\beta}
\end{array}
\right)\to0.
\]
\end{remark}

%% file: section6-main-results.tex
\section{Main result}
\label{sec:main-results}

We now combine the collision-free estimate with the bound on stable block systems.

\begin{theorem}[Minimality with high probability]
\label{thm:minimality}
Consider the finite-catalog random Moore automaton with \(\boldsymbol\delta\sim\operatorname{Unif}(\operatorname{Tra}_{n,m})\) and iid labels \(\boldsymbol\xi(q)\sim\rho\), independent of \(\boldsymbol\delta\). Assume
\[
k_\ast:=\min_j |S_j|\geq3
\text{ and }
\sum_{j=1}^D\rho_j^2<1.
\]
Then \(\mathbb P(|\boldsymbol A_\infty(n)|>0)\to0\). In other words, the quotient of the full \(n\)-state transition system by the state congruence has \(n\) equivalence classes with probability tending to one.
\end{theorem}

\begin{proof}
Choose \(0<\beta<1-2/k_\ast\). Since \(k_\ast\geq3\), such a \(\beta\) exists and satisfies \(\beta<1/2\). By Corollary~\ref{cor:mesoscopic-entry}, there exists a deterministic time \(t_n=O(\log\log n)\) such that
\[
\mathbb P\left(|\boldsymbol A_{t_n}|>n^{1+\beta}\right)\to0.
\]
Since the pruning process is decreasing, \(\boldsymbol A_\infty(n)\subseteq \boldsymbol A_{t_n}\). Therefore
\[
\begin{aligned}
\mathbb P(|\boldsymbol A_\infty(n)|>0)
&\leq
\mathbb P\left(|\boldsymbol A_{t_n}|>n^{1+\beta}\right) \\
&\quad+
\mathbb P\left(0<|\boldsymbol A_\infty(n)|\leq n^{1+\beta}\right).
\end{aligned}
\]
The first term tends to zero. On the event \(\{0<|\boldsymbol A_\infty(n)|\leq n^{1+\beta}\}\), Lemma~\ref{lem:Ainfty-gives-stable-block-system} produces a stable block system \(\boldsymbol{\mathcal B}_\infty\) satisfying
\[
1\leq h(\boldsymbol{\mathcal B}_\infty)
=|\boldsymbol A_\infty(n)|
\leq n^{1+\beta}.
\]
Hence this event is contained in \(E_n(\beta)\). By Lemma~\ref{lem:small-obstruction-bound}, \(\mathbb P(E_n(\beta))\to0\). Consequently \(\mathbb P(|\boldsymbol A_\infty(n)|>0)\to0\).
\end{proof}

\begin{corollary}[Reachable part]
\label{cor:reachable-part}
Assume the hypotheses of Theorem~\ref{thm:minimality}. Let
\[
    \operatorname{acc}(\boldsymbol{\mathcal A}_n)
    :=
    \{\boldsymbol\delta^\ast(1,u):u\in\Sigma^\ast\}
\]
be the reachable part of \(\boldsymbol{\mathcal A}_n\). Then
\[
\mathbb P\left(\operatorname{acc}(\boldsymbol{\mathcal A}_n)\text{ is minimal}\right)\to1.
\]
\end{corollary}

\begin{proof}
If two reachable states are equivalent in the restricted automaton, then they are also equivalent as states of the full automaton.
\end{proof}

\begin{remark}[Accessible model]
Theorem~\ref{thm:minimality} is stated for \(\operatorname{Tra}_{n,m}\), where transition values are independent. The uniformly accessible model is a natural next target, but the proof above relies on this independence in both the local exploration estimate and the obstruction count. The corollary nevertheless shows that the actually reachable part is minimal with high probability.
\end{remark}

%% file: section7-conclusion.tex
\section{Concluding remarks}
\label{sec:conclusion}

We proved that in the finite-catalog model the prefix-dependent congruence is trivial with high probability. Thus the full \(n\)-state random transition system is minimal with respect to the congruence considered here, provided the catalog law is non-degenerate and every catalog type admits at least three admissible symbols. In particular, the reachable automaton is minimal with high probability.

The proof separates two effects. The early pruning dynamics is governed by a collision-free recursion and uses only local tree-likeness of finite explorations. The limiting set requires a different argument: it is generated by congruence classes, and therefore any nonempty limiting set gives a stable block system. A first-moment bound excludes such obstructions up to the scale reached by the pruning process when \(k_\ast\geq3\).

The main limitation of our result is that the final stable-block-system bound is proved only in the iid uniform transition model. The argument suggests, however, that uniformity may not be essential: the local exploration and first-moment estimates should remain valid for independent, possibly non-identically distributed transitions. Extending the result to such inhomogeneous transition models is a natural next step. Other extensions include weakening the condition \(k_\ast\geq3\) and incorporating growing catalog diversity in the critical case \(k_\ast=2\).